\documentclass[12pt, leqno]{amsart}

\usepackage{mathtools}
\usepackage{graphicx}
\usepackage{amssymb,amsmath,amsthm,amscd}
\usepackage{mathrsfs}
\usepackage{lscape}
\usepackage{enumerate}
\usepackage[usenames,dvipsnames]{color}
\usepackage[colorlinks=true, pdfstartview=FitV, linkcolor=blue,citecolor=blue,urlcolor=blue]{hyperref}
\usepackage{comment}
\usepackage{tikz}
\usepackage[all]{xy}

\allowdisplaybreaks[3]

\numberwithin{equation}{section}

\newenvironment{red}{\relax\color{red}}{\relax}
\newenvironment{blue}{\relax\color{blue}}{\hspace*{.5ex}\relax}

\newcommand{\ber}{\begin{red}}
\newcommand{\er}{\end{red}}
\newcommand{\beb}{\begin{blue}}
\newcommand{\eb}{\end{blue}}

\theoremstyle{plain}
\newtheorem{lemma}{Lemma}[section]
\newtheorem{proposition}[lemma]{Proposition}
\newtheorem{theorem}[lemma]{Theorem}

\theoremstyle{definition}
\newtheorem{remark}[lemma]{Remark}

\newlength{\mylength}
\title{Lax-Sato formulation of the Novikov-Veselov Hierarchy}

\author{Sylvain Carpentier}

\date{\today}

\begin{document}

\begin{abstract}We construct a hierarchy of pairwise commuting flows $d/dt_{i,n}$ indexed by $i \in \{1,2 \}$ and $n \in \mathbb{Z}_{\geq 0}$ on triples $(\mathcal{L}_1, \mathcal{L}_2, \mathcal{H})$ where $\partial_1$ and $\partial_2$ are two commuting derivations, $\partial_i \mathcal{L}_i$ is a self-adjoint pseudodifferential operator in $\partial_i$ and $\mathcal{H}$ is the formal Schr\"{o}dinger operator $\mathcal{H}=\partial_1 \partial_2 +u$. $\mathcal{L}_1, \mathcal{L}_2$ and $\mathcal{H}$ are coupled by the relations $\mathcal{H} \mathcal{L}_i+\mathcal{L}_i^* \mathcal{H}=0$. We show that the flows $d/dt_{1,n}+d/dt_{2,n}$ commute with the involution $(\mathcal{L}_1, \mathcal{L}_2, \mathcal{H}) \mapsto (\mathcal{L}_2, \mathcal{L}_1, \mathcal{H})$ and that the first equation of this reduced hierarchy is the Novikov-Veselov equation.
\end{abstract}

\maketitle

\section{Introduction}

The Novikov-Veselov equation (NV) at energy level $E \in \mathbb{R}$ is the $(2+1)$-dimensional evolution PDE
\begin{equation} \label{NVeq}
\frac{du}{dt}=\partial_z^3u+\partial_{\bar{z}}^3 u+\partial_z((u-E)v)+\partial_{\bar{z}}((u-E) \bar{v}), \hspace{1 cm}
3 \,\partial_z u =\partial_{\bar{z}} v,
\end{equation}
where $u(x,y,t)$ is a real-valued function, $v(x,y,t)$ is complex-valued, $\partial_z=\frac{1}{2}(\partial_x+i \partial_y)$ and $\partial_{\bar{z}}=\frac{1}{2}(\partial_x-i \partial_y)$. It was derived in [NV] as a compatibility condition between differential equations satisfied by a certain multiparametric generalized Baker-Akhiezer function. The discovery of the algebro-geometric data defining uniquely this generalized Baker-Akhiezer function was based on the earlier work [DKN]. One can see that the NV equation reduces to the Korteweg-de Vries equation (KdV) when $u$ and $v$ are assumed to only depend on the spatial variable $x$.  A feature shared by the KdV and the NV equations is the property to be integrable by inverse scattering. We refer the reader to the review paper [G] for a description of this method. Since 1984 where it appeared for the first time, the NV equation has been widely studied from an analytic standpoint. For instance, fairly recently ([N]) R. Novikov proved that at positive energy $E >0$, the NV equation does not admit exponentially localized solitons. 
\\
 \indent From an algebraic point of view, the value of the energy $E$ does not matter as one can substitute $u$ for $u-E$, hence in the sequel we will assume that $E=0$. The NV equation was originally introduced in [NV] as part of a hierarchy of evolution equations cast in the form of Manakov L-A-B triples ([M])
\begin{equation*}
\frac{d\tilde{\mathcal{H}}}{dt_n}=[\tilde{\mathcal{H}},A_n+\bar{A_n}]+(B_n+\bar{B_n})\tilde{\mathcal{H}}, \, \, \, n \in \mathbb{Z}_{\geq 0},
\end{equation*}
where $\tilde{\mathcal{H}}$ is the two-dimensional Schr\"{o}dinger operator $\partial_z \partial_{\bar{z}}+u$ and $A_n, B_n$ are differential operators. The NV equation corresponds to $n=1$, where $A_1=\partial_z^3+v \partial_z$ and $B_1$ is the operator of multiplication by the function $v_z$.
\\
\indent In this paper, we complete this hierarchy using a two-dimensional Lax-Sato formalism, involving pseudodifferential operators in two commuting derivations $\partial_1$ and $\partial_2$. In a first step, we construct a family of pairwise commuting flows $d/dt_{i,n}$ indexed by $i \in \{1,2\}$ and $n \in \mathbb{Z}_{\geq 0}$ on the space of triples $(\mathcal{L}_1, \mathcal{L}_2, \mathcal{H})$ where $\mathcal{L}_1$ is a formal pseudodifferential operator in  $\partial_1$ such that $\partial_1 \mathcal{L}_1$ is self-adjoint with leading term $\partial_1^2$, $\mathcal{L}_2$ is a formal pseudodifferential operator in  $\partial_2$ such that also $\partial_2 \mathcal{L}_2$ is self-adjoint with leading term $\partial_2^2$, $\mathcal{H}=\partial_1 \partial_2+u$ for some function $u$ and the following relations hold
\begin{equation*}
\mathcal{H} \mathcal{L}_i+\mathcal{L}_i^* \mathcal{H}=0, \, \, \, i \in \{1,2 \}.
\end{equation*}
The restriction of the flows $d/dt_{i,n}, \, n \in \mathbb{Z}_{\geq 0}$ to the pseudodifferential operator $\mathcal{L}_i$ is the BKP hierarchy, hence our construction couples two copies of BKP via the two-dimensional formal Schr\"{o}dinger operator $\partial_1 \partial_2+u$.
In a second step, we show that the flows $d/dt_{1,n}+d/dt_{2,n}$ commute with the involution $(\mathcal{L}_1, \mathcal{L}_2, \mathcal{H}) \mapsto (\mathcal{L}_2, \mathcal{L}_1, \mathcal{H})$ for all $n \in \mathbb{Z}_{\geq 0}$. The NV equation is finally retrieved as the first non-trivial equation in this reduced hierarchy. 
\\
\\
\textit{Acknowledgements.} The author was supported by a Junior Fellow award from the Simons Foundation. He is grateful to Igor Krichever for suggesting this problem and many fruitful discussions. This paper was inspired by Grushevsky and Krichever's work on a difference-differential analogue of the Novikov-Veselov hierarchy in [GK].

\section{pseudodifferential operators}
We recall briefly the definition and some elementary algebraic properties of differential operators over an algebra  endowed with two commuting derivations $\partial_1$ and $\partial_2$. Let $\mathcal{V}$ be a commutative associative algebra over $\mathbb{C}$. Let $\partial_1$ and $\partial_2$ be two commuting $\mathbb{C}$-linear derivations of $\mathcal{V}$. We assume that the algebra $\mathcal{V}$ is a domain. The algebra of differential operators over $\mathcal{V}$ is the space $\mathcal{V}[\partial_1, \partial_2]$ where the multiplication is defined by the relations $\partial_i v=v \partial_i+ \partial_i(v)$ for any $v \in \mathcal{V}$ and $i \in \{1,2 \}$, and $\partial_1 \partial_2=\partial_2 \partial_1$. It embeds in the larger algebras $\mathcal{V}[\partial_2]((\partial_1^{-1}))$ and $\mathcal{V}[\partial_1]((\partial_2^{-1}))$, where the multiplication is defined by
\begin{equation*}
\partial_i^{-1} v=v\partial_i^{-1}-\partial_i(v)\partial_i^{-2}+\partial_i^2(v) \partial_i^{-3}-... \, \, , \, \, v \in \mathcal{V}, \, \, i \in \{1,2\}.
\end{equation*}
The elements of $\mathcal{V}[\partial_2]((\partial_1^{-1}))$ (resp. $\mathcal{V}[\partial_1]((\partial_2^{-1}))$) are called pseudodifferential operators in $\partial_1$ (resp. $\partial_2$) over $\mathcal{V}[\partial_2]$ (resp. $\mathcal{V}[\partial_1]$). 
We define the negative and positive part of a pseudodifferential operator (in $\partial_1$) $\mathcal{P}=\underset{n \leq N}\sum p_n \partial_1^{n} \in \mathcal{V}[\partial_2]((\partial_1^{-1}))$ as follows:
\begin{equation*}
\mathcal{P}_+=\underset{n \geq 0 } \sum p_n \partial_1^n, \, \, \, \, \mathcal{P}_-=\underset{n < 0 } \sum p_n \partial_1^n, \, \, \, \, \mathcal{P}=\mathcal{P}_+ + \mathcal{P}_-  \, .
\end{equation*}
 The positive and negative parts of elements in $\mathcal{V}[\partial_1]((\partial_2^{-1}))$ are defined symmetrically. The adjunction operation on pseudodifferential operators is the unique linear morphism defined by the properties
\begin{equation*}
\partial_i^*=-\partial_i^*, \, \, \,\,  \, v^*=v, \, \,  \,\, \,  (PQ)^*=Q^*P^*,
\end{equation*}
where $i \in \{1,2\}$, $v \in \mathcal{V}$ and $P,Q$ are any two pseudodifferential operators. For any $a \in \mathcal{V}$ we introduce the self-adjoint differential operator 
\begin{equation*}
\mathcal{H}_a=\partial_1\partial_2+a.
\end{equation*}

\begin{lemma} \label{one}
Let $a \in \mathcal{V}$. For any pseudodifferential operator $\mathcal{P} \in \mathcal{V}[\partial_2]((\partial_1^{-1}))$ there exists a unique pseudodifferential operator $\mathcal{Q} \in \mathcal{V}[\partial_2]((\partial_1^{-1}))$ such that $\mathcal{P}-\mathcal{Q} \mathcal{H}_a \in \mathcal{V}((\partial_1^{-1}))$.
\end{lemma}
\begin{proof}
We prove the existence by induction on the degree of $\mathcal{P}$ as a polynomial in $\partial_2$. If this degree is $0$, one can take $\mathcal{Q}=0$. If it is $N >0$, let $\mathcal{P}_N \partial_2^N$ be its top degree component. Then $\mathcal{P}-\mathcal{P}_N\partial_2^N-\mathcal{P}_N \partial_2^{N-1}\partial_1^{-1}a=\mathcal{P}-\mathcal{P}_N\partial_2^{N-1} \partial_1^{-1} \mathcal{H}_a$ has degree at most $N-1$. By the induction hypothesis there exists $\mathcal{R} \in \mathcal{V}((\partial_1^{-1}))[\partial_2]$ such that 
$\mathcal{P}-\mathcal{P}_N\partial_2^{N-1} \partial_1^{-1} \mathcal{H}_a-\mathcal{R} \mathcal{H}_a \in \mathcal{V}((\partial_1^{-1}))$. Hence one can let $\mathcal{Q}=\mathcal{R}+\mathcal{P}_N\partial_2^{N-1} \partial_1^{-1}$. As for the unicity, one needs to prove that if $\mathcal{P} \in \mathcal{V}((\partial_1^{-1}))$ and $\mathcal{Q} \in \mathcal{V}[\partial_2]((\partial_1^{-1}))$ are such that $\mathcal{P}=\mathcal{Q} \mathcal{H}_a$, then $\mathcal{P}=\mathcal{Q}=0$. This follows from looking at the top degrees of both sides as polynomials in $\partial_2$, since $\mathcal{V}$ is a domain.
\end{proof}
Note that the same statement holds after swapping $\partial_1$ with $\partial_2$. Finally, the (differential) order of a pseudodifferential operator is the grading of its top graded component where both $\partial_1$ and $\partial_2$ have grading $1$ and elements of $\mathcal{V}$ have grading $0$.

\section{A coupled BKP hierarchy}
Let $\mathcal{A}$ be the algebra of differential polynomials over $\mathbb{C}$ generated by the elements $u$, $(v_i)_{i \in \mathbb{Z}_{\geq 0}}$, $(w_j)_{j \in \mathbb{Z}_{\geq 0}}$ and their jets $\partial_1^{a} \partial_2^b(u), \partial_1^c \partial_2^d(v_i), \partial_1^{e} \partial_2^f(w_j), \, a,b,c,d,e,f \in \mathbb{Z}_{\geq 0}$ for two commuting derivations $\partial_1$, $\partial_2$, subject to the relations
\begin{equation}\label{defrel}
\partial_2(\mathcal{L}_1)=[\mathcal{L}_1, \partial_1^{-1}u], \, \, \, \partial_1(\mathcal{L}_2)=[\mathcal{L}_2, \partial_2^{-1}u],
\end{equation}
where 
\begin{equation*}
\begin{split}
\mathcal{L}_1&=\partial_1^{-1}(\partial_1^2+v_0+\partial_1^{-1}v_1 \partial_1^{-1}+\partial_1^{-2}v_2 \partial_1^{-2}+ ...),
\\
\mathcal{L}_2&=\partial_2^{-1}(\partial_2^2+w_0+\partial_2^{-1}w_1 \partial_2^{-1}+\partial_2^{-2}w_2 \partial_2^{-2}+ ...).
\end{split}
\end{equation*}
The LHS of equations \eqref{defrel} should be understood as follows
\begin{equation*}
\begin{split}
\partial_2(\mathcal{L}_1)&=\partial_1^{-1}(\partial_1^2+\partial_2(v_0)+\partial_1^{-1}\partial_2(v_1) \partial_1^{-1}+\partial_1^{-2}\partial_2(v_2) \partial_1^{-2}+ ...),
\\
\partial_1(\mathcal{L}_2)&=\partial_2^{-1}(\partial_2^2+\partial_1(w_0)+\partial_2^{-1}\partial_1(w_1) \partial_2^{-1}+\partial_2^{-2}\partial_1(w_2) \partial_2^{-2}+ ...).
\end{split}
\end{equation*}
Note that both $\partial_1\mathcal{L}_1$ and $\partial_2 \mathcal{L}_2$ are self-adjoint by construction. The relations \eqref{defrel} are well-defined since both $\partial_1[\mathcal{L}_1,\partial_1^{-1}u]$ and $\partial_2[\mathcal{L}_2,\partial_2^{-1}u]$ are self-adjoint pseudodifferential operators of order at most $0$. Indeed, for all $i \in \{1,2 \}$, 
\begin{equation*}
\begin{split}
(\partial_i[\mathcal{L}_i,\partial_i^{-1}u])^*&=(\partial_i \mathcal{L}_i \partial_i^{-1}u-u \partial_i^{-1} \partial_i \mathcal{L}_i)^* \\
                      &=-u\partial_i^{-1} (\partial_i\mathcal{L}_i)^* +(\partial_i\mathcal{L}_i)^* \partial_i^{-1} u \\
                      &=-u \mathcal{L}_i + \partial_i \mathcal{L}_i \partial_i^{-1}u\\
                      &=\partial_i[\mathcal{L}_1,\partial_i^{-1}u].
\end{split}
 \end{equation*}
  Explicitely, we have 
\begin{equation*}
\mathcal{A}=\mathbb{C}[\partial_1^n(v_m), \partial_2^k(w_l), \partial_1^p \partial_2^q(u)| k,l,m,n,p,q \in \mathbb{Z}_{\geq 0}],
\end{equation*}
 and the $\partial_2$ (resp. $\partial_1$) jets of the $v_i$'s (resp. $w_i$'s) can be expressed in terms of their $\partial_1$ (resp. $\partial_2$) jets and of $u$. In particular, the two first terms in \eqref{defrel} give
\begin{equation} \label{rela}
\begin{split}
\partial_2(v_0)&=\partial_1(u), \, \, \partial_2(v_1)=\partial_1(u)v_0-u \partial_1(v_0), \\
\partial_1(w_0)&=\partial_2(u), \, \, \partial_1(w_1)=\partial_2(u)w_0-u \partial_2(w_0).
\end{split}
\end{equation}
The subfields of constants for $\partial_1$ and $\partial_2$ in $\mathcal{A}$ are both equal to $\mathbb{C}$. We have  $\partial_2(\mathcal{L}_1)=\partial_2 \mathcal{L}_1- \mathcal{L}_1 \partial_2$ and similarly $\partial_1(\mathcal{L}_2)=\partial_1 \mathcal{L}_2- \mathcal{L}_2 \partial_1$, hence
the two equations \eqref{defrel} can be rewritten in the equivalent form
\begin{equation} \label{defrel2}
[\mathcal{L}_1, \partial_2+ \partial_1^{-1} u]=0, \, \, \, \, [\mathcal{L}_2, \partial_1+\partial_2^{-1}u]=0,
\end{equation}
which can be recast, using relations $\partial_i \mathcal{L}_i+\mathcal{L}_i^* \partial_i=0$  where $i \in \{1,2 \}$, in the form
\begin{equation} \label{defrel3}
\mathcal{H} \mathcal{L}_1=-\mathcal{L}_1^* \mathcal{H}, \, \, \, \, \mathcal{H} \mathcal{L}_2=-\mathcal{L}_2^* \mathcal{H},
\end{equation}
where $\mathcal{H}$ is the formal two-dimensional Schr\"{o}dinger operator 
\begin{equation*}
\mathcal{H}=\partial_1 \partial_2+u. 
\end{equation*}
 For all $n \in \mathbb{Z}_{\geq 0}$ and $i \in \{1,2\}$, we let 
 \begin{equation*}
 A_{i,n}=(\mathcal{L}_i^{2n+1})_+ \in \mathcal{A}[\partial_i].
 \end{equation*}
  For instance, we have
\begin{equation} \label{As}
A_{1,0}=\partial_1, \, \, \, A_{2,1}=\partial_2^3+3w_0 \partial_2, \, \, \, A_{1,2}= \partial_1^5+5v_0 \partial_1^3+5 \partial_1(v_0) \partial_1^2+(5\partial_1^2(v_0)+5v_1+10 v_0^2)\partial_1.
 \end{equation}
   By Lemma \ref{one}, for all $n \in \mathbb{Z}_{\geq 0}$, there exists a unique decomposition 
\begin{equation*}
A_{1,n}=B_{2,n}+C_{2,n} \mathcal{H},
\end{equation*}
where $B_{2,n} \in \mathcal{A}((\partial_2^{-1}))$ and $C_{2,n} \in \mathcal{A}[\partial_1]((\partial_2^{-1}))$. Note that for all $n \in \mathbb{Z}_{\geq 0}$ the order of $B_{2,n}$ is negative. We define $B_{1,n}  \in \mathcal{A}((\partial_1^{-1}))$ and $C_{1,n}  \in \mathcal{A}[\partial_2]((\partial_1^{-1}))$ uniquely for all $n \in \mathbb{Z}_{\geq 0}$ in a symmetric way:
\begin{equation*}
A_{2,n}=B_{1,n}+C_{1,n} \mathcal{H}.
\end{equation*}
For instance, we have
\begin{equation*}
B_{2,0}=-\partial_2^{-1}u, \, \, \, B_{1,1}=-\partial_1^{-1}(\partial_2^2(u)+3uw_0)+\partial_1^{-1}u\partial_1^{-1} \partial_2(u)-\partial_1^{-1} \partial_2(u) \partial_1^{-1} u-\partial_1^{-1}u\partial_1^{-1}u\partial_1^{-1}u.
\end{equation*}

\begin{lemma}  \label{two}
For all $n \in \mathbb{Z}_{\geq 0}$ and $i \in \{1,2 \}$,  $A_{i,n} \partial_i^{-1}$ is a self-adjoint differential operator in $\partial_i$.
\end{lemma}
\begin{proof}
 Since $\partial_i \mathcal{L}_i=-\mathcal{L}_i^* \partial_i$,  we have $\partial_i \mathcal{L}_i^{2n+1}=-(\mathcal{L}_i^{2n+1})^* \partial_i$ for all $n \in \mathbb{Z}_{\geq 0}$. Hence $\partial_i \mathcal{L}_i^{2n+1}$ is a self-adjoint pseudodifferential operator in $\partial_i$ and can be written as
 \begin{equation*}
 \partial_i \mathcal{L}_i^{2n+1}=\partial_i^{2n+2}+\partial_i^n a_{i,n} \partial_i^n+...+ \partial_i a_{i,1} \partial_i+a_{i,0}+ \partial_i^{-1} a_{i,-1} \partial_i^{-1}+...
 \end{equation*}
 for some elements $a_{i,k} \in \mathcal{A}$, where $k \in \mathbb{Z}_{ \leq n}$. It follows that 
  \begin{equation*}
A_{i,n}\partial_i^{-1}=\partial_i^{2n}+\partial_i^{n-1} a_{i,n} \partial_i^{n-1}+...+ a_{i,1},
 \end{equation*}
 proving the lemma.
\end{proof}

\begin{lemma} \label{three}
For all $i \in \{1,2 \}$ and $n \in \mathbb{Z}_{\geq 0}$ 
\begin{equation} \label{AH}
\mathcal{H} A_{i,n}+A^*_{i,n} \mathcal{H}=A^*_{i,n}(u).
\end{equation}
\end{lemma}
\begin{proof}
From the identity \eqref{defrel2}, which can be rewritten as $[\mathcal{L}_i, \partial_i^{-1} \mathcal{H}]=0$, it follows that $[\mathcal{L}_i^{2n+1}, \partial_i^{-1} \mathcal{H}]=0$. Therefore,
 \begin{equation*}
 \begin{split}
 0&=[\mathcal{L}_i^{2n+1}, \partial_i^{-1} \mathcal{H}]_+ \\
  &=[A_{i,n}, \partial_i^{-1} \mathcal{H}]_+ \\
  &=A_{i,n} \partial_i^{-1} \mathcal{H}-\partial_i^{-1} \mathcal{H} A_{i,n} +(\partial_i^{-1} \mathcal{H} A_{i,n})_{-} \\
  &=A_{i,n} \partial_i^{-1} \mathcal{H}-\partial_i^{-1} \mathcal{H} A_{i,n} +(\partial_i^{-1}u A_{i,n})_{-} \\
  &=A_{i,n} \partial_i^{-1} \mathcal{H}-\partial_i^{-1} \mathcal{H} A_{i,n} +\partial_i^{-1}A_{i,n}^*(u)\\
  &=\partial_i^{-1}(A_{i,n}^*(u)-A_{i,n}^* \mathcal{H}- \mathcal{H} A_{i,n}).
 \end{split}
 \end{equation*}
 To obtain the third line in this system of equations we used Lemma \ref{two}, and to deduce the fifth line from the fourth we used the fact that any differential operator $\mathcal{P}$ in $\partial_i$ is equal modulo the right ideal $\partial_i \mathcal{A}[\partial_i]$ to $\mathcal{P}^*(1)$.
 \end{proof}

\begin{lemma} \label{four}
For all $n \in \mathbb{Z}_{\geq 0}$ and $i \in \{1,2\}$, $\partial_i B_{i,n}$ is a self-adjoint pseudodifferential operator in $\partial_i$.
\end{lemma}
\begin{proof} By symmetry of the construction we only need to prove the statement for $i=1$.  Let $X=\partial_1 B_{1,n}$. The following equalities hold modulo right multiplication by $\mathcal{H}$ in the algebra $\mathcal{A}[\partial_1]((\partial_2^{-1}))$. By definition of the pseudodifferential operator $B_{1,n}$  we have 
 \begin{equation*}
 \partial_1 A_{2,n}=X \, \, \text{mod} \, \, \mathcal{H}.
 \end{equation*}
 By Lemma \ref{three}, we get
 \begin{equation*}
 \begin{split}
 A_{2,n}^*(u)&=\mathcal{H} A_{2,n} \, \, \text{mod} \, \, \mathcal{H} \\
  &=\mathcal{H} \partial_1^{-1} X \, \, \text{mod} \, \, \mathcal{H} \\
 &=\partial_2 X+u \partial_1^{-1} X \, \, \text{mod} \, \, \mathcal{H} \\
                    &= \partial_2(X) + X \partial_2 + u \partial_1^{-1} X \, \, \text{mod} \, \, \mathcal{H} \\
                    &=\partial_2(X) +u \partial_1^{-1} X-X \partial_1^{-1} u  \, \, \text{mod} \, \, \mathcal{H}.
 \end{split}
 \end{equation*}
 Since both sides of the equality do not depend on $\partial_2$, one can remove $\text{mod} \, \, \mathcal{H}$ by Lemma \ref{one} and get:
 \begin{equation*}
 A_{2,n}^*(u)=\partial_2(X) +u \partial_1^{-1} X-X \partial_1^{-1} u.
 \end{equation*}
 After taking the adjoint of this equation, we see that $X-X^*$ must satisfy the differential equation
 \begin{equation} \label{key}
 \partial_2(X-X^*)=(X-X^*) \partial_1^{-1}u-u \partial_1^{-1}(X-X^*),
 \end{equation}
 from which it follows that $X=X^*$. Indeed, the coefficients of $X$ as a pseudodifferential operator in $\partial_1$ are differential polynomials in the generators of $\mathcal{A}$ with no constant part and as we noted earlier the subfield of constants for $\partial_2$ in $\mathcal{A}$ is $\mathbb{C}$. Hence, if we assume that $X-X^*$ is nonzero, its differential order is the same as the differential order of $\partial_2(X-X^*)$, which gives a contradiction using equation \eqref{key}.
\end{proof}

We recall that an evolutionary derivation of $\mathcal{A}$ is a derivation which commutes with both $\partial_1$ and $\partial_2$. In particular, it is uniquely defined by the values it takes on the generators of $\mathcal{A}$, or equivalently by the values it takes on $\mathcal{L}_1$, $\mathcal{L}_2$ and $\mathcal{H}$.

\begin{proposition}\label{pro}
The following formulas induce a well-defined family $(d/dt_{i,n})_{i \in \{1,2\}, n \in \mathbb{Z}_{\geq 0}}$ of evolutionary derivations of $(\mathcal{A},\partial_1,\partial_2)$:
\begin{equation*}
\begin{split}
\frac{d \mathcal{L}_1}{dt_{1,n}}&=[A_{1,n}, \mathcal{L}_1], \, \, \, \frac{d \mathcal{L}_2}{dt_{2,n}}=[A_{2,n}, \mathcal{L}_2], \\
\frac{d \mathcal{L}_1}{dt_{2,n}}&=[A_{2,n}, \mathcal{L}_1] \, \, \text{mod} \, \, \mathcal{H}=[B_{1,n}, \mathcal{L}_1],\\
\frac{d \mathcal{L}_2}{dt_{1,n}}&=[A_{1,n}, \mathcal{L}_2] \, \, \text{mod} \, \, \mathcal{H}=[B_{2,n}, \mathcal{L}_2], \\
\frac{d \mathcal{H}}{dt_{1,n}}&=[A_{1,n}, \mathcal{H}] \, \, \text{mod} \, \, \mathcal{H}=-A_{1,n}^*(u), \\
\frac{d \mathcal{H}}{dt_{2,n}}&=[A_{2,n}, \mathcal{H}] \, \, \text{mod} \, \, \mathcal{H}=-A_{2,n}^*(u).
\end{split}
\end{equation*}
\end{proposition}
\begin{proof}
We only need to check that these formulas preserve the form of $\mathcal{L}_1$, $\mathcal{L}_2$ and $\mathcal{H}$. In other words, we need to verify that the image of $\mathcal{H}$ for these derivations is an element of $\mathcal{A}$, and that the images of $\partial_i \mathcal{L}_i, i \in \{1,2\}$ are self-adjoint pseudodifferential operators in $\partial_i$ of order at most $0$. For $\mathcal{H}$ this is obvious since $-A_{i,n}^*(u)$ is an element of $\mathcal{A}$. The fact that it is equal to $[A_{i,n}, \mathcal{H}]$ modulo $\mathcal{H}$ is an immediate consequence of Lemma \ref{three}. By Lemma \ref{two} and Lemma \ref{four}, it follows that both $\partial_i[A_{i,n}, \mathcal{L}_i]$ and $\partial_i[B_{i,n}, \mathcal{L}_i]$ are self-adjoint for all $i \in \{1,2 \}$ and $n \in \mathbb{Z}_{\geq 0}$. 
Indeed,
\begin{equation*}
\begin{split}
(\partial_i[A_{i,n},\mathcal{L}_i])^*&=(\partial_i A_{i,n} \partial_i^{-1} \partial_i \mathcal{L}_i-\partial_i \mathcal{L}_i\partial_i^{-1} \partial_i A_{i,n})^* \\
                      &=-(\partial_i \mathcal{L}_i)^* \partial_i^{-1} (\partial_i A_{i,n})^*+(\partial_i A_{i,n})^* \partial_i^{-1} (\partial_i \mathcal{L}_i)^* \\
                      &=-\partial_i \mathcal{L}_i A_{i,n}+\partial_i A_{i,n} \mathcal{L}_i\\
                      &=\partial_i[A_{i,n},\mathcal{L}_i],
\end{split}
\end{equation*}
and similarly for $\partial_i[B_{i,n},\mathcal{L}_i]$.
Moreover, for all $n \in \mathbb{Z}_{\geq 0}$ and $i \in \{1,2\}$, both $\partial_i[A_{i,n},\mathcal{L}_i]$ and $\partial_i[B_{i,n},\mathcal{L}_i]$ have orders at most $0$, since the $B_{i,n}$'s have negative orders and $[A_{i,n}, \mathcal{L}_i]=-[(\mathcal{L}_i^{2n+1})_-, \mathcal{L}_i]$. Finally, $[A_{2,n}, \mathcal{L}_1] \, \, \text{mod} \, \, \mathcal{H}=[B_{1,n}, \mathcal{L}_1]$ since $\mathcal{H}\mathcal{L}_1=-\mathcal{L}_1^* \mathcal{H}= 0 \, \, \text{mod} \, \, \mathcal{H}$, and similarly after swapping $\partial_1$ with $\partial_2$.
\end{proof}
\begin{remark}
For $i \in \{1,2 \}$, the evolution of $\mathcal{L}_i$ along the flows $d/dt_{i,n}, n \in \mathbb{Z}_{\geq 0}$ is by definition the BKP hierarchy (see [DJKM]). Hence in our construction we have two copies of BKP,  coupled via the operator $\mathcal{H}$.
\end{remark}
The derivations $d/dt_{1,0}$ and $d/dt_{2,0}$ identify with $\partial_1$ and $\partial_2$. We give the evolutions of $u$, $v_0$ and $w_0$ under $d/dt_{1,1}$ and $d/dt_{2,1}$, obtained using Proposition \ref{pro} and equation \eqref{As}:
\begin{equation*}
\begin{split}
\frac{du}{dt_{1,1}}&=\partial_1^3(u)+3\partial_1(v_0 u), \, \, \, \, \,  \frac{du}{dt_{2,1}}=\partial_2^3(u)+3 \partial_2(w_0u), \\
\frac{d v_0}{dt_{1,1}}&= \partial_1^3(v_0)+6v_0 \partial_1(v_0)+3 \partial_1(v_1) , \, \, \, \, \, \frac{dv_0}{dt_{2,1}}= \partial_2^3(v_0)+3 \partial_1(w_0 u), \\
\frac{d w_0}{dt_{1,1}}&= \partial_1^3(w_0)+3 \partial_2(v_0 u), \, \, \, \, \, \frac{d w_0}{dt_{2,1}}=\partial_2^3(w_0)+6w_0 \partial_2(w_0)+3\partial_2(w_1).
\end{split}
\end{equation*}

We are now going to prove that the derivations $d/dt_{i,n}$ are pairwise compatible. In order to do so, we first state an auxiliary lemma.

\begin{lemma} \label{five}
Let $n, m \in \mathbb{Z}_{\geq 0}$ and $i, j \in \{1,2\}$. Then there exists a differential operator $R_{n,m}^{i,j}$ in $\mathcal{A}[\partial_1, \partial_2]$ such that 
\begin{equation}\label{zerocurvature}
\frac{d A_{i,n}}{dt_{j,m}}-\frac{d A_{j,m}}{dt_{i,n}}+[A_{i,n},A_{j,m}]=R_{n,m}^{i,j} \mathcal{H}.
\end{equation}
Moreover, the operator $R_{n,m}^{i,j}$ is skew-adjoint and is identically $0$ when $i=j$. 
\end{lemma}
\begin{proof}
When $i=j$ this statement is standard. Indeed,
\begin{equation*}
\begin{split}
0&=[\mathcal{L}_i^{2m+1},\mathcal{L}_i^{2n+1}]_+ \\
 &=[A_{i,m},(\mathcal{L}_i^{2n+1})_-]_+ +[(\mathcal{L}_i^{2m+1})_-, A_{i,n}]_+ + [A_{i,m}, A_{i,n}] \\
 &=[A_{i,m},\mathcal{L}_i^{2n+1}]_+ +[\mathcal{L}_i^{2m+1}, A_{i,n}]_+ - [A_{i,m}, A_{i,n}] \\
 &=\frac{d A_{i,n}}{dt_{i,m}}-\frac{d A_{i,m}}{dt_{i,n}}+[A_{i,n},A_{i,m}].
\end{split}
 \end{equation*}
  It is enough to prove the Lemma when $i=1$ and $j=2$, after which the case $i=2$ and $j=1$ follows by symmetry.
It is clear that there exists a unique decomposition of the differential operator $[A_{1,n},A_{2,m}] \in \mathcal{A}[\partial_1, \partial_2]$ of the form
\begin{equation} \label{dec}
[A_{1,n},A_{2,m}]=P\partial_1+Q \partial_2+a+R_{n,m} \mathcal{H}
\end{equation}
where $P \in \mathcal{A}[\partial_1]$, $Q \in \mathcal{A}[\partial_2]$, $a \in \mathcal{A}$ and $R_{n,m} \in \mathcal{A}[\partial_1, \partial_2]$. By definition, we have
\begin{equation*}
\frac{d \mathcal{L}_1}{dt_{2,m}}=[A_{2,m}, \mathcal{L}_1] \, \, \text{mod} \, \, \mathcal{H}.
\end{equation*}
Since $\mathcal{L}_1 \mathcal{H}=-\mathcal{L}_1^* \mathcal{H}$ we have $\mathcal{H} \mathcal{L}_1^k=0 \, \, \text{mod} \, \, \mathcal{H}$ for any $k \in \mathbb{Z}_{\geq 0}$, hence 
\begin{equation*}
\frac{d \mathcal{L}_1^{2n+1}}{dt_{2,m}}=[A_{2,m}, \mathcal{L}_1^{2n+1}] \, \, \text{mod} \, \, \mathcal{H}.
\end{equation*}
It is straightforward to check that any element of $\partial_1^{-1}\mathcal{A}[\partial_2][[\partial_1^{-1}]]$ is equal modulo $\mathcal{H}$ to a pseudodifferential operator in $\partial_1$ of negative degree. Therefore
\begin{equation*}
\frac{d A_{1,n}}{dt_{2,m}}=[A_{2,m}, A_{1,n}]_+ \, \, \text{mod} \, \, \mathcal{H}.
\end{equation*}
In the decomposition \eqref{dec} of $[A_{1,n}, A_{2,m}]$, the part $Q \partial_2$ is equal modulo $\mathcal{H}$  to a pseudodifferential operator in $\partial_1$ of negative degree, since it is a differential operator in $\partial_2$ with no order zero term. We deduce that 
\begin{equation*}
\frac{d A_{1,n}}{dt_{2,m}}=-P \partial_1 -a.
\end{equation*}
This equation implies that $a=0$, since $A_{1,n}$ does not have a zero order coefficient. Similarly, one can prove that 
\begin{equation*}
\frac{d A_{2,m}}{dt_{1,n}}=Q \partial_2, 
\end{equation*}
from which we conclude that 
\begin{equation*}
\frac{d A_{1,n}}{dt_{2,m}}-\frac{d A_{2,m}}{dt_{1,n}}+[A_{1,n},A_{2,m}]=R_{n,m} \mathcal{H}.
\end{equation*}
We are left to prove that $R_{n,m}$ is skew-adjoint. In order to do so, we use the identities $\mathcal{H}A_{i,k}+A_{i,k}^* \mathcal{H}=A_{i,k}^*(u)$, valid for all $i \in \{1,2\}$ and all $k \in \mathbb{Z}_{\geq 0}$ (see Lemma \ref{three}). We have
\begin{equation*}
\begin{split}
\mathcal{H}[A_{1,n},A_{2,m}]&=(A_{1,n}^*(u)-A_{1,n}^*\mathcal{H})A_{2,m}- (A_{2,m}^*(u)-A_{2,m}^*\mathcal{H})A_{1,n} \\
    &= [A_{1,n}^*,A_{2,m}^*]\mathcal{H}+A_{1,n}^*(u)A_{2,m}+A_{2,m}^*A_{1,n}^*(u)-A_{2,m}^*(u)A_{1,n}-A_{1,n}^*A_{2,m}^*(u) \\
     &=-[A_{1,n},A_{2,m}]^*\mathcal{H}+A_{1,n}^*(u)A_{2,m}+A_{2,m}^*A_{1,n}^*(u)-A_{2,m}^*(u)A_{1,n}-A_{1,n}^*A_{2,m}^*(u).
\end{split}
\end{equation*}
Combining this equation with \eqref{dec}, we get 
\begin{equation*}
\begin{split}
&\mathcal{H}(P\partial_1+ Q \partial_2)+\mathcal{H}R_{n,m} \mathcal{H}-(\partial_1 P^*+\partial_2 Q^*) \mathcal{H}+\mathcal{H}R_{n,m}^* \mathcal{H}= \\
&A_{1,n}^*(u)A_{2,m}+A_{2,m}^*A_{1,n}^*(u)-A_{2,m}^*(u)A_{1,n}-A_{1,n}^*A_{2,m}^*(u).
\end{split}
\end{equation*}
In particular, $\mathcal{H}(R_{n,m}+R_{n,m}^*)\mathcal{H}$ must be in the space $\mathcal{A}[\partial_1] \oplus \mathcal{A}[\partial_1] \partial_2 \oplus \mathcal{A}[\partial_2] \oplus \mathcal{A}[\partial_2] \partial_1$. This can only be if $R_{n,m}+R_{n,m}^*=0$, since $\mathcal{A}$ is a domain.
\end{proof}

\begin{theorem}
The evolutionary derivations $d/dt_{i,n}, \, i \in \{1,2\}, \, n \in \mathbb{Z}_{\geq 0}$ of the differential algebra $(\mathcal{A}, \partial_1, \partial_2)$ pairwise commute.
\end{theorem}
\begin{proof}
Since the commutator of two evolutionary derivations is an evolutionary derivation, it is enough to show that these derivations pairwise commute on the generators of $\mathcal{A}$. We first check that $\frac{d^2 \mathcal{H}}{dt_{i,n} dt_{j,m}}=\frac{d^2 \mathcal{H}}{dt_{j,m} dt_{i,n}}$ for all $ i,j \in \{1,2\}, \, n,m \in \mathbb{Z}_{\geq 0}$. By definition of $d/dt_{i,n}$ and Lemma \ref{three} we have 
\begin{equation*}
\frac{d \mathcal{H}}{dt_{i,n}}=-A_{i,n}^*(u)=[A_{i,n}, \mathcal{H}]-(A_{i,n}+A_{i,n}^*)\mathcal{H}.
\end{equation*}
Applying the derivation $d/dt_{j,m}$ to this equation we get
\begin{equation*}
\begin{split}
\frac{d^2 \mathcal{H}}{dt_{j,m}dt_{i,n}}&=[\frac{dA_{i,n}}{dt_{j,m}}, \mathcal{H}]-(\frac{dA_{i,n}}{dt_{j,m}}+\frac{dA^*_{i,n}}{dt_{j,m}})\mathcal{H}\\
&+[A_{i,n}, \frac{d\mathcal{H}}{dt_{j,m}}]-(A_{i,n}+A_{i,n}^*)\frac{d\mathcal{H}}{dt_{j,m}}\\
&=[\frac{dA_{i,n}}{dt_{j,m}}, \mathcal{H}]+[A_{i,n},[A_{j,m}, \mathcal{H} ]]+ (A_{j,m}+{A^*_{j,m}})\mathcal{H}A_{i,n} +(A_{i,n}+A^*_{i,n}) \mathcal{H} A_{j,m} \\
&+((A_{i,n}+A^*_{i,n})(A_{j,m}+A^*_{j,m})-A_{i,n}(A_{j,m}+A^*_{j,m})-(A_{i,n}+A^*_{i,n})A_{j,m}-\frac{dA_{i,n}}{dt_{j,m}}-\frac{dA^*_{i,n}}{dt_{j,m}}) \mathcal{H}.
\end{split}
\end{equation*}
After a straightforward computation we see that
\begin{equation*}
\frac{d^2 \mathcal{H}}{dt_{j,m}dt_{i,n}}-\frac{d^2 \mathcal{H}}{dt_{i,n}dt_{j,m}}=[X,\mathcal{H}]-(X+X^*)\mathcal{H},
\end{equation*}
where $X=\frac{d A_{i,n}}{dt_{j,m}}-\frac{d A_{j,m}}{dt_{i,n}}+[A_{i,n},A_{j,m}]$. Using Lemma \ref{five} we deduce that
\begin{equation*}
\frac{d^2 \mathcal{H}}{dt_{j,m}dt_{i,n}}-\frac{d^2 \mathcal{H}}{dt_{i,n}dt_{j,m}}=-\mathcal{H} R_{n,m}^{i,j}\mathcal{H}-\mathcal{H}(R_{n,m}^{i,j})^* \mathcal{H}=0,
\end{equation*}
since $R_{n,m}^{i,j}$ is skew-adjoint. We now prove that $\frac{d^2 \mathcal{L}_1}{dt_{i,n} dt_{j,m}}=\frac{d^2 \mathcal{L}_1}{dt_{j,m} dt_{i,n}}$ for all $ i,j \in \{1,2\}, \, n,m \in \mathbb{Z}_{\geq 0}$.  By definition of $d/dt_{i,n}$ there exists an element $P_{i,n} \in \mathcal{A}[\partial_2]((\partial_1^{-1}))$ such that
\begin{equation*}
\frac{d \mathcal{L}_1}{dt_{i,n}}=[A_{i,n}, \mathcal{L}_1]+P_{i,n} \mathcal{H}.
\end{equation*}
 Note that $P_{1,n}=0$ for all $n \in \mathbb{Z}_{\geq 0}$. The following equalities hold modulo left multiplication by $\mathcal{H}$ in $\mathcal{A}[\partial_2]((\partial_1^{-1}))$
\begin{equation*}
\begin{split}
\frac{d^2 \mathcal{L}_1}{dt_{i,n}dt_{j,m}}&=[\frac{d A_{i,n}}{dt_{j,m}}, \mathcal{L}_1]+[A_{i,n},[A_{j,m}, \mathcal{L}_1]]-P_{j,m} \mathcal{H} A_{i,n}+P_{i,n} \frac{d \mathcal{H}}{dt_{j,m}} \, \, \, \text{mod} \, \, \mathcal{H} \\
&= [\frac{d A_{i,n}}{dt_{j,m}}, \mathcal{L}_1]+[A_{i,n},[A_{j,m}, \mathcal{L}_1]]-P_{j,m} \mathcal{H} A_{i,n}-P_{i,n} \mathcal{H} A_{j,m} \, \, \, \text{mod} \, \, \mathcal{H}.
\end{split}
\end{equation*}
Hence, by Lemma \ref{five} we have
\begin{equation*}
\begin{split}
\frac{d^2 \mathcal{L}_1}{dt_{i,n}dt_{j,m}}-\frac{d^2 \mathcal{L}_1}{dt_{j,m}dt_{i,n}}&=[\frac{d A_{i,n}}{dt_{j,m}}-\frac{d A_{j,m}}{dt_{i,n}}+[A_{i,n},A_{j,m}], \mathcal{L}_1]] \, \, \, \text{mod} \mathcal{H} \\
&=[R_{n,m}^{i,j} \mathcal{H}, \mathcal{L}_1] \, \, \, \text{mod} \, \mathcal{H} \\
&= R_{n,m}^{i,j} \mathcal{H} \mathcal{L}_1 \, \, \, \text{mod} \, \mathcal{H} \\
&= -R_{n,m}^{i,j} \mathcal{L}_1^* \mathcal{H} \, \, \, \text{mod} \,  \mathcal{H} \\
&=0 \, \, \, \text{mod} \, \mathcal{H}.
\end{split}
\end{equation*}
The LHS is a multiple of $\mathcal{H}$ in $\mathcal{A}[\partial_2]((\partial_1^{-1}))$. But it is also in $\mathcal{A}((\partial_1^{-1}))$. Therefore it is identically $0$ by Lemma \ref{one}. The proof that $\frac{d^2 \mathcal{L}_2}{dt_{i,n} dt_{j,m}}=\frac{d^2 \mathcal{L}_2}{dt_{j,m} dt_{i,n}}$ for all $ i,j \in \{1,2\}, \, n,m \in \mathbb{Z}_{\geq 0}$ is similar.
\end{proof}

\section{Real reduction}
We are now going to reduce the hierarchy constructed in the previous section, under the following involution $\tau$ of the $\mathbb{R}$-algebra $\mathcal{A}$:
\begin{equation*}
\begin{split}
z & \mapsto \bar{z}, \\
\partial_1^n(v_i)& \mapsto \partial_2^n(w_i), \\
\partial_2^n(w_j)& \mapsto \partial_1^n(v_j) , \, \, \, \, \, \, \, \, \text{for all} \, i,j,n,p,q \in \mathbb{Z}_{\geq 0}, \, \, z \in \mathbb{C}.\\
\partial_1^p \partial_2^q(u)& \mapsto \partial_1^q \partial_2^p(u).
\end{split}
\end{equation*}
It is immediate that $\tau(\partial_1(\tau(\mathcal{L}_2)))=\partial_2(\mathcal{L}_2)$. The fact that $\tau(\partial_1(\tau(\mathcal{L}_1)))=\partial_2(\mathcal{L}_1)$ follows from the relation \ref{defrel}. Hence the two derivations of $\mathcal{A}$, $\partial_2$ and $\tau \partial_1 \tau$, coincide. We can then extend the involution $\tau$ to the algebra of pseudodifferential operators over $\mathcal{A}$ by the formula $\tau(\mathcal{P})= \tau \mathcal{P} \tau$, which corresponds to sending $\partial_1$ to $\partial_2$ and vice versa. In particular, for all $a \in \mathcal{A}$ and $i \in \{1,2\}$, we have
\begin{equation*}
\tau(\partial_i(a))=\tau(\partial_i)(\tau(a)).
\end{equation*}
Moreover, $u$ is an invariant element of $\mathcal{A}$ for the involution $\tau$.

\begin{proposition} \label{propo}
The (pairwise commuting evolutionary) derivations $d/dt_n:=d/dt_{1,n}+d/dt_{2,n}$ of $\mathcal{A}$ are invariant under the involution $\tau$. In other words, for  all $a \in \mathcal{A}$ and $n \in \mathbb{Z}_{\geq 0}$ one has $\tau({\frac{da}{dt_n}})=\frac{d \tau(a)}{dt_n}$.
\end{proposition}
\begin{proof}
One only needs to check that this property holds on the generators of $\mathcal{A}$, which are the coefficients of $\mathcal{L}_1$, $\mathcal{L}_2$ and $\mathcal{H}$. This follows from the equalities
\begin{equation*}
\begin{split}
A_{2,n}&=\tau(A_{1,n}), \, \, \, \, \, B_{2,n}=\tau(B_{1,n}), \\
\mathcal{L}_2&= \tau(\mathcal{L}_1), \, \, \, \, \, \mathcal{H}=\tau(\mathcal{H}),
\end{split}
 \end{equation*} 
 valid for all $n \in \mathbb{Z}_{\geq 0}$.
\end{proof}
Explicitly, the derivations $d/dt_n$ are defined by the formulas
\begin{equation*}
\begin{split}
\frac{d \mathcal{H}}{dt_n}&=[A_{1,n} +\tau(A_{1,n}), \mathcal{H}] \, \, \text{mod} \, \, \mathcal{H}=-A_{1,n}^*(u)-\tau(A_{1,n})^{*}(u)\\
\frac{d \mathcal{L}_1}{dt_n}&=[A_{1,n} +\tau(A_{1,n}), \mathcal{L}_1] \, \, \text{mod} \, \, \mathcal{H}=[A_{1,n} +B_{1,n}, \mathcal{L}_1].
\end{split}
\end{equation*}
We have $d\mathcal{L}_2/dt_n=\tau(d \mathcal{L}_1/dt_n)$ for all $n \in \mathbb{Z}_{\geq 0}$, by Proposition \ref{propo}.
The evolution of $u$ and $v_0$ for the first nontrivial derivation $d/dt_1$ in this reduced hierarchy is given by
\begin{equation} \label{NVag}
\begin{split}
\frac{du}{dt_1}&=(\partial_1^3+\tau(\partial_1)^3+3 \partial_1 v_0+3 \tau(\partial_1) \tau(v_0))(u), 
\\
\frac{dv_0}{dt_1}&=(\partial_1^3+\tau(\partial_1)^3)(v_0)+6 v_0 \partial_1(v_0)+3 \partial_1(u \tau(v_0))+3 \partial_1(v_1).
\end{split}
\end{equation}
Note that the first coeffient in the relation \eqref{rela} rewrites as $\tau(\partial_1)(v_0)=\partial_1(u)$. Hence, one retrieves the Novikov-Veselov equation \eqref{NVeq} from \eqref{NVag} after letting $v=3 v_0$. One can interpret u and $v_n, n \in \mathbb{Z}_{\geq 0}$ as being complex-valued functions of two real variables $x$ and $y$, $\partial_1$ the derivation $1/2(\partial_x+i \partial_y)$ and the involution as being the complex conjugation. From the condition $u=\tau(u)$ it follows that $u$ is a real-valued function.

\end{document}